\documentclass[12pt, a4paper]{amsart}
\usepackage{amsmath, amssymb, amsthm, graphicx}
\usepackage[english]{babel}
\usepackage[pdfhighlight=/P, colorlinks=true, linkcolor=blue, citecolor=blue, urlcolor=blue]{hyperref}

\usepackage[left=2.5cm, right=2.5cm, top=2.5cm, bottom=2.5cm]{geometry}

\DeclareMathOperator{\ad}{ad}
\DeclareMathOperator{\const}{const}
\DeclareMathOperator{\diag}{diag}
\DeclareMathOperator{\sech}{sech}
\DeclareMathOperator{\tr}{tr}
\DeclareMathOperator{\dn}{dn}
\DeclareMathOperator{\sn}{sn}
\DeclareMathOperator{\GL}{GL}
\DeclareMathOperator{\gl}{gl}
\newcommand{\CS}{\begin{pmatrix} \cos\chi & \sin\chi \\ \sin\chi & -\cos\chi \end{pmatrix}}
\newcommand{\CSF}{\begin{pmatrix} \cos\chi_1 & \sin\chi_1 \\ \sin\chi_1 & -\cos\chi_1 \end{pmatrix}}
\newcommand{\CSX}{\begin{pmatrix} \cos 2x & \sin 2x \\ \sin 2x & -\cos 2x \end{pmatrix}}

\theoremstyle{plain}
\newtheorem{theorem}{Theorem}
\newtheorem{lemma}{Lemma}
\newtheorem{prop}{Proposition}
\theoremstyle{definition}

\theoremstyle{remark}
\newtheorem{remark}{Remark}

\title{Spinning top in quadratic potential and matrix dressing chain}

\author{V.E. Adler}
\address{Landau Institute for Theoretical Physics, 
1A Semenov pr., 142432 Chernogolovka, Russia}
\email{adler@itp.ac.ru}

\author{A.P. Veselov}
\address{Department of Mathematical Sciences,
Loughborough University, Loughborough LE11 3TU, UK}
\email{A.P.Veselov@lboro.ac.uk}


\begin{document}
\begin{abstract}
We show that the equations of motion of the rigid body about centre of mass in the Newtonian field with a quadratic potential are special reductions of period-one closure of the Darboux dressing chain for the Schr\"odinger operators with matrix potentials. We show that the corresponding matrix Schr\"odinger operators are maximally finite-gap (in the sense that for all sufficiently large energies all solutions of the corresponding Schr\"odinger equation are bounded) and describe their spectrum explicitly. The general $2\times 2$-matrix case of the dressing chain, providing also some exotic matrix versions of the harmonic oscillator, is discussed in more detail.  \end{abstract}

\maketitle
\rightline{\em Dedicated to V.V. Kozlov on his 75th birthday}
\bigskip

\section{Introduction}

The study of the motion of a rigid body about a fixed point is one of the most classical problems of mechanics. The equations of the free top were written by Euler and were integrated in the newly invented elliptic functions by Jacobi \cite{Jacobi}. In constant gravitational field the system in general is non-integrable with the exception of two famous Lagrange and Kovalevskaya cases \cite{Kozlov1983}. The integration of the equations of motion in the Kovalevskaya case \cite{Kov} was the climax of the theory of integrable systems in XIX century.

We will be interested in the case of the motion of the top about its centre of mass in the Newtonian field with quadratic potential. In the axial-symmetric case the integrability of the system was shown by de Brun \cite{Brun}. The corresponding system coincides with the Clebsch integrable case of motion of the free rigid body in the infinite fluid \cite{Clebsch}. A close relation of the de Brun case to the Jacobi problem of geodesics on ellipsoid was observed by Kozlov \cite{Kozlov1}.

A remarkable development in 1980s was due to Reyman \cite{Reiman} and Bogoyavlenskij \cite{Bogo84}, who independently discovered that this system turns out to be integrable also for general quadratic potential, see details in the next section.

The aim of this paper is to explain the link of this problem with the Darboux dressing chain for the matrix Schr\"odinger operators. In the scalar case such chain was first considered by Infeld and Hull \cite{IH}, applications to spectral theory were initiated by A.B. Shabat \cite{Shabat}. The periodic closures with relation to the algebro-geometric finite-gap theory and hierarchies of Painlev\'e equations were studied in detail in \cite{Adler94, VS}. 

We will show that in the matrix case already the period-one closure gives rise to an interesting system on $d\times d$ (in general, complex) matrices $F(x)$ and $B(x)$:
\begin{equation}\label{DC}
 CF'+F'C=[C,F^2+B]+2\alpha C,\quad B'=[B,F]  
\end{equation}
where $F'=dF/dx$, constant matrix $C$ and $\alpha \in \mathbb C$ are parameters.
When $\alpha=0$ we get a matrix integrable system of Dubrovin type \cite{Dubrovin, Dubrovin2} and show that its reduction describes the motion of the $d$-dimensional top in general quadratic potential. In that sense the matrix dressing chain (\ref{DC}) can be considered as natural $GL(d)$ extension of this important system on $SO(d)$. We would like to mention that the corresponding system on $SO(d)$ has multi-parameter integrable generalisations \cite{Bogo84, Reiman}, in the free case discovered by Manakov \cite{Manakov}, so the fact that it is precisely the rigid body system, which appears in this relation, seems to be remarkable.

The dressing chain interpretation makes link with the spectral theory of Schr\"odinger operators with matrix potential.\footnote{ A different link of the rigid body dynamics with matrix Schr\"odinger equation was mentioned in \cite{V2}.} In particular, the corresponding Schr\"odinger operators 
$$
L=-D_x^2+U(x), \quad U=-D_xF+F^2+B,\quad D_x=\frac{d}{dx},
$$
are finite-gap in Dubrovin's sense \cite{Dubrovin}. For the Bogoyavlenskij top, corresponding to  $\alpha=0$ and the reduction $C=C^\top, \, F=-F^\top, \, B=B^\top,$ we show that the operators $L$ are self-adjoint with the spectrum described by a suitable real version of the spectral curve of the classical mechanical system. In particular, we show that the corresponding Schr\"odinger operators are maximally finite-gap  in the sense that for sufficiently large energies all solutions of the corresponding Schr\"odinger equation are bounded (cf. \cite{GS}).

However, from the point of view of the spectral theory of matrix Schr\"odinger operator more natural is the reduction
$
 B=0,\,\,  F=F^\top,\,\, CC^\top=I,
$
leading directly to the symmetric potentials $U=U^\top.$
In the simplest non-trivial $2\times 2$-case with $\alpha=0$ we have a family of $\pi$-periodic Schr\"odinger operators with Mathieu-like matrix potentials
$$
 U=A\CSX,
$$
with the spectral problem explicitly solvable in elementary functions. This is quite interesting since in the scalar case there are no such non-constant periodic finite-gap operators.

When $\alpha\neq 0$ we have the following isospectral family of exotic harmonic oscillators with
$$
 U=\alpha^2x^2I + \beta x\CS,\quad \chi=\gamma x^2+ \delta,
$$
solvable in terms of Weber (parabolic cylinder) functions.

We show that the general solutions of equation (\ref{DC}) in $2\times 2$-case can be expressed in either elliptic functions, or certain Painlev\'e II and IV transcendents and discuss the spectral properties of the corresponding Schr\"odinger operators in Section 6.

We conclude with a brief discussion of the related matrix KdV hierarchy and corresponding matrix Novikov equations.

\section{Motion of the top in Newtonian field with quadratic potential}

We follow here Bogoyavlenskij \cite{Bogo92} to derive the equations of motion of the rigid body about its centre of mass in Newtonian field with quadratic potential. As a concrete example one can consider the oscillations of a satellite about its centre of mass \cite{Bel}.

Let us start with the problem of the rotation of a 3-dimensional rigid body $B$ around its centre of mass $O$ in the Newtonian central field of a material point $P$ having mass $m$. If the distance $R$ from $O$ to the centre of the field is much larger than the size of the body, the equations of satellite motion in the leading approximation order have the form \cite{Bel}
\begin{equation}\label{Bel}
\dot M=M\times\Omega+\gamma m R^{-3} p\times Jp, \quad \dot p=p\times\Omega.
\end{equation}
Here 
$
 J_{ij}:=\int_B \rho(x)x_ix_j\,dx
$
is the inertia tensor of the body, $\Omega$ is the angular velocity, $M=J\Omega$ is the angular momentum, $p$ is the unit vector in moving frame directed from the centre of mass of the body to the centre of the field and cross denotes vector product. A remarkable fact is that these equations coincide with the Clebsch integrable case for the Kirchhoff equations \cite{Clebsch}, which was rediscovered in this context by de Brun \cite{Brun}.

The equations should be modified for any shape $E$ of attracting body (e.g. Earth) as
\begin{equation}\label{Bog1}
\dot M=M\times \Omega+\gamma \int_E R^{-3}(y) \rho_E(y) p(y)\times Jp(y)\,dy, \quad \dot p(y)=p(y)\times\Omega,
\end{equation}
where $p(y)$ is the unit vector in moving frame directed from $y\in E$ to the centre of mass $O$ and $R(y)$ is the distance from $y$ to $O$ (see \cite{Bogo92}).

Let us identify now the Euclidean space $\mathbb R^3$ with vector product and the space of $3\times 3$ skew-symmetric matrices with the Lie bracket, and introduce the symmetric matrix $P$ with entries
\[
P_{ij}=\gamma \int_E R^{-3}(y) \rho_E(y) p_i(y)p_j(y)\,dy.
\]
Then the equations of motion take a remarkable closed matrix form \cite{Bogo92}
\begin{equation}\label{Bog2}
 \dot M=[M,\Omega]-[P, J], \quad \dot P=[P, \Omega]
\end{equation}
where $M=J\Omega+\Omega J\in so(3)$ and we used the identity
$
 [X, YX+XY]=[X^2,Y]
$
for any two matrices $X$ and $Y$.

In this form the equations can be naturally extended to dimension $d$ with $M,\Omega\in so(d)$, $P$ and (fixed) $J$ are symmetric $d\times d$ matrices \cite{A}. The corresponding equations (\ref{Bog2}) are the Euler equations on Lie algebra 
$
 \mathfrak g=so(d) \ltimes symm(d),
$
which is the semi-direct product of $so(d)$ with the space of $d\times d$ symmetric matrices considered as abelian Lie algebra, with the Hamiltonian
\[
 H=\tr\frac{1}{2} M\Omega - \tr P J.
\]
When $d=3$ and $P_{ij}=p_ip_j$ has rank 1, the system  (\ref{Bog2}) reduces to the Clebsch case of the Kirchhoff equations known to be Euler equation on the Lie algebra $E(3)=so(3) \ltimes \mathbb R^3$ (see \cite{NS}). The integrability of its $d$-dimensional version was proved by Perelomov \cite{Perelomov}, who provided the Lax representation of the Clebsch system depending on spectral parameter.

Reyman \cite{Reiman} and Bogoyavlenskij \cite{Bogo84} proved the integrability of the general $d$-dimensional system (\ref{Bog2}) in a similar way by re-writing it in the Lax form 
$\dot{\mathcal L}=[\mathcal  L, \mathcal A]$ with
\begin{equation}\label{laxbog}
 \mathcal  L=P+\mu M-\mu^2J^2, \quad \mathcal A=\Omega-\mu J,
\end{equation}
$\mu$ is spectral parameter (see also \cite{Bogo92} and more detailed analysis in \cite{RSTS}). The integrals are the coefficients of the characteristic equation of $\mathcal L$
\begin{equation}\label{laxspec}
\det (\mathcal L-\lambda I)= \det (P+\mu M-\mu^2J^2-\lambda I)=0,
\end{equation}
which determines the corresponding algebraic spectral curve. The motion is linearised on the corresponding Prym variety and the solutions can be expressed in terms of the Prym theta functions, see \cite{Bogo92}. One can extend this integration to the motion of the body frame on the group $SO(d),$ where the system is Lagrangian with the Lagrangian
\begin{equation}\label{Lagran}
\frac{1}{2}\tr J(Q^{-1}\dot Q)^2-\frac{1}{2}\tr JQ^{-1}A Q, \quad Q\in SO(d)
\end{equation}
with fixed symmetric $J, A$, see \cite{Bogo92}.

We will show now that a natural extension of this system to the group $GL(d)$ appears in the theory of matrix dressing chain as the period-one closure of the matrix dressing chain. 

\section{Matrix Darboux transformations and dressing chain}

The inverse scattering problems for matrix Schr\"odinger operators have been studied by many authors (see e.g. \cite{CD, AO, O, WK}) in relation with integrable nonlinear PDEs, starting with the pioneering paper by Lax \cite{Lax}, who introduced the matrix analogue of the celebrated Korteweg-de Vries equation. In particular, the Darboux transformations (DT) were used in \cite{GV, Gon} to describe matrix Schr\"odinger equations with trivial monodromy and some matrix KdV solitons.

Note that in contrast to the scalar case, the matrix DT are not necessarily related to the factorization of Schr\"odinger operator \cite{EGR}. Our approach is close to the treatment of DT by Suzko \cite{Suzko_2005}.

Consider the one-dimensional Schr\"odinger operator 
\[
 L=-D_x^2+U(x),\quad D_x=\frac{d}{dx},\quad x\in\mathbb R
\]
with $U(x)\in \gl_d({\mathbb C})$ being at this stage any $d\times d$ matrix-valued function.

Let $\psi=(\psi^1,\dots,\psi^d)^\top$ be any solution of the corresponding matrix Schr\"odinger equation
\begin{equation}\label{psixx}
 \psi''=(U-\lambda)\psi,\quad U=U(x).
\end{equation}
We say that the map $\tilde\psi=\psi'+F\psi$, $F\in\gl_d({\mathbb C})$ is an elementary Darboux transformation of the operator $L$ if $\tilde\psi$ satisfies another copy of Schr\"odinger equation with some new potential $\tilde U$. It is easy to check that this requirement is equivalent to relations
\begin{gather}
\label{u1}
 \tilde U=U+2F',\\
\label{uf}
 U'+F''=[U,F]+2F'F.
\end{gather}
The first one is just the definition of the new potential, while the second one defines $F$ as a solution of some ODE. In the scalar case $d=1$ this ODE is integrated to Riccati equation which is then linearized back to Schr\"odinger equation:
\[
 u=-f'+f^2+\beta,\quad \beta=\const
 \qquad\xrightarrow{~\displaystyle{f=-\phi'/\phi}~}\qquad
 \phi''=(u-\beta)\phi,
\]
so that the construction of Darboux transformation is based on a particular solution of Schr\"odinger equation at a particular value of spectral parameter $\lambda=\beta$.

In the matrix case equation (\ref{uf}) is linearizable as well, but in a bit more complicated way. It is convenient to rewrite it as a system, introducing the new variable $B$ as follows:
\begin{equation}\label{ufb}
 U=-F'+F^2+B,\quad B'=[B,F].
\end{equation}
The substitution
\begin{equation}\label{phiB}
 F=-\phi'\phi^{-1},\quad B=\phi\Lambda\phi^{-1},\quad
 \phi\in\GL_d({\mathbb C}),\quad \Lambda\in\gl_d({\mathbb C})
\end{equation}
brings this system to equations
\begin{equation}\label{phiU}
 \phi''=U\phi-\phi\Lambda,\quad \Lambda'=0,
\end{equation}
or, equivalently, as
\begin{equation}\label{phiU2}
L\phi=\phi\Lambda.
\end{equation}
The matrix $\Lambda$ can be taken in Jordan form without loss of generality due to the change $\phi\to\phi C$, $\Lambda\to C^{-1}\Lambda C$. If $\Lambda$ is diagonal matrix, $\Lambda=\diag(\beta^{(1)},\dots,\beta^{(d)})$, then construction of DT is based on $d$ particular solution of Schr\"odinger equation $L \phi^{(i)} = \beta^{(i)}\phi^{(i)}$ for the column vectors $\phi^{(i)}$ of $\phi$ corresponding to the values of spectral parameter $\lambda=\beta^{(i)}$. In particular, if $\Lambda$ is scalar matrix, $\Lambda=\beta I$, then we need $d$ linearly independent solutions of Schr\"odinger equation at $\lambda=\beta$. In the case of Jordan blocks the adjoint vectors are needed.

In the operator language, Darboux transformation is equivalent to relations
\begin{equation}\label{factorization}
 L=\bar AA+B,\quad \tilde L=A\bar A+B,\quad [A,B]=0 \quad\Rightarrow\quad AL=\tilde LA
\end{equation}
where
\[
 L=-D^2_x+U,\quad A=D_x+F,\quad \bar A=-D_x+F.
\]
In scalar case (or in the case $B=\beta I$) the condition $[A,B]=0$ is equivalent to $B=\const$, so that DT can be introduced in terms of factorization of Schr\"odinger operator $L-\beta$.

The higher order Darboux transformations are defined by intertwining operators of the form $A=D^m_x+a_1D^{m-1}_x+\dots+a_m$. The common feature of the scalar and matrix cases is that any such DT can be represented as a composition of elementary DT \cite{EGR,GV}, that is there exists a sequence of operators
\begin{equation}\label{ALLA}
  L_n=-D^2_x+U_n,\quad A_n=D_x+F_n,\quad A_nL_n=L_{n+1}A_n
\end{equation}
such that $L=L_0$ and $\tilde L=L_m$.

Another thing which should be mentioned is that, in general, the matrix Darboux transformation is not symmetric, that is operators $L$ and $\tilde L$ are not on the equal footing. Indeed, it is easy to find that the inverse of DT is given by formula $\psi=(B-\lambda I)^{-1}\bar A\tilde\psi$ and the factor $B-\lambda I$ may be dropped out only if $B$ is scalar matrix.

Finally, in the matrix case DT can be combined with the conjugation by a constant matrix $\tilde\psi=C\psi$, $\tilde U=CUC^{-1}$. 

We will see now that this trivial transformation becomes important in the study of the periodic closure of the following {\it matrix dressing chain} governing the iterations of Darboux transformations defined by sequence of operators (\ref{ALLA}). For this sequence, equations (\ref{u1}), (\ref{ufb}) take the form
\begin{equation}\label{ufbn}
 U_{n+1}=U_n+2F'_n,\quad U_n=-F'_n+F^2_n+B_n,\quad B'_n=[B_n,F_n].
\end{equation}
Elimination of potentials $U$ brings to the matrix dressing chain in the form
\begin{equation}\label{fbx}
 F'_{n+1}+F'_n=F^2_{n+1}-F^2_n+B_{n+1}-B_n,\quad B'_n=[B_n,F_n].
\end{equation}
Another convenient form is
\begin{equation}\label{vbx}
 V'_{n+1}+V'_n=(V_{n+1}-V_n)^2+B_n,\quad B'_n=[B_n,V_{n+1}-V_n]
\end{equation}
where variables $V$ are introduced by equations
\begin{equation}\label{ufv}
 U_n=2V'_n,\quad F_n=V_{n+1}-V_n.
\end{equation}

The Lax pair for the lattice (\ref{fbx}) is obtained directly from the definition of DT:
\begin{equation}\label{DLA}
 T\psi_n=(D_x+F_n)\psi_n,\quad
 D^2_x\psi_n=(U_n-\lambda)\psi_n,\quad U_n=F^2_n-F'_n+B_n
\end{equation}
where $T$ is the shift operator $n\mapsto n+1$. Elimination of the derivatives from the second equation brings this pair to an equivalent difference form:
\begin{equation}\label{TLA}
 D_x\psi_n=(T-F_n)\psi_n, \quad T^2\psi_n=(TF_n+F_nT+B_n-\lambda I)\psi_n.
\end{equation}
These linear problems can be equivalently rewritten in the matrix form for $2d$-vector $\Psi$, and the consistency condition takes form of zero curvature representation:
\begin{equation}\label{FU}
 \Psi'_n=\hat U_n\Psi_n,\quad \Psi_{n+1}=\hat F_n\Psi_n \quad\Rightarrow\quad
 \hat F'_n=\hat U_{n+1}\hat F_n-\hat F_n\hat U_n.
\end{equation}
The chain (\ref{fbx}) corresponds to the block matrices
\begin{equation}\label{UFf}
 \hat U_n=\begin{pmatrix} 0 & I\\ U_n-\lambda I & 0\end{pmatrix},\qquad
 \hat F_n=\begin{pmatrix} F_n & I\\ F^2_n+B_n-\lambda I & F_n\end{pmatrix}
\end{equation}
and the chain (\ref{vbx}) to the matrices
\begin{equation}\label{UFv}
 \hat U_n=\begin{pmatrix} V_n & I\\
    V'_n-V^2_n-\lambda I & -V_n\end{pmatrix},\qquad
 \hat F_n=\begin{pmatrix} V_{n+1} & I \\
    B_n-V_nV_{n+1}-\lambda I & -V_n\end{pmatrix}.
\end{equation}

\section{Periodic closure of matrix dressing chain and spinning top}

In the scalar case the periodic closures of the dressing chain
\[
 f'_{n+1}+f'_n=f^2_{n+1}-f^2_n+\beta_{n+1}-\beta_n,\quad f_{n+N}\equiv f_n,~ \beta_{n+N}=\beta_n
\]
were studied in detail in \cite{VS}, where it was proved, in particular, that for odd period $N=2g+1$ this is an integrable Hamiltonian system and the corresponding Schr\"odinger operators are finite-gap and all finite-gap operators can be obtained in this way. In this sense, this system is an alternative form of the stationary KdV equations from the pioneering Novikov's work \cite{Nov}. Another result
obtained in \cite{VS} is that the more general closure with $\beta_{n+N}=\beta_n+2\alpha$ leads to the hierarchy of higher analogues of Painlev\'e IV equation, describing ``finite-gap" deformations of the harmonic oscillator corresponding to $N=1$, when we have
$
 f'=\alpha,~ f=\alpha x,~ u=\alpha^2x^2-\alpha.
$

In the matrix case we consider the following periodic closure of the dressing chain (\ref{fbx}) adding the conjugation by a constant matrix $C$:
\begin{equation}\label{perN}
 F_{n+N}=CF_nC^{-1},\quad B_{n+N}=CB_nC^{-1}+2\alpha I.
\end{equation}

In particular, when $N=1$ the chain (\ref{fbx}) turns into the matrix ODE
\begin{equation}\label{fbtop}
  CF'+F'C=[C,F^2+B]+2\alpha C,\quad B'=[B,F],
\end{equation}
which is quite interesting and is the main object of our study. In fact, the chain with arbitrary period $N$ is equivalent to its $d\times d$--block reduction
\[
 F=\diag(F_1,\dots,F_N),\quad
 B=\diag(B_1,\dots,B_N),\quad
 C\mapsto\begin{pmatrix}
    0 & 1 & & \\
      & \ddots & \ddots & \\
      & & \ddots & 1\\
    C & & & 0
   \end{pmatrix}.
\]
It should be noted that in the scalar case the algebraic properties of the dressing chain, as well as the analytic properties of its solutions, depend essentially on the parity of period $N$ \cite{VS}. Indeed, if $N$ is even then the system cannot be solved with respect to derivatives and an additional constraint
\[
 f^2_N+\beta_N-f^2_{N-1}-\beta_{N-1}+\dots-f^2_1-\beta_1=0
\]
appears. In the matrix case an analogous situation may occur for any $N$, depending on the choice of matrix $C$. 
From now on we will assume that the eigenvalues $c_1,\dots c_d$ of matrix $C$ satisfy the relation
\begin{equation}\label{nondeg}
c_i+c_j\neq 0
\end{equation}
for all $i,j=1,\dots,d$, which guarantees the invertibility of the map $X \to CX+XC.$

Replacing $T\psi$ with $\mu C\psi$  in (\ref{TLA}), we can rewrite equation (\ref{fbtop}) in the form
\begin{equation}\label{PQfb}
 \mathcal L'=[\mathcal L, \mathcal A]-2\mu\alpha C
\end{equation}
where 
\begin{equation}\label{lax}
 \mathcal L=\mu^2C^2-\mu(CF+FC)-B,\quad \mathcal A=F-\mu C.
\end{equation}
In particular, for $\alpha=0$ we have the Lax representation for the corresponding system: 
\begin{equation}\label{lax2}
 \mathcal L'=[\mathcal L, \mathcal A],
\end{equation}
or, from the block form (\ref{UFf}) with 
$U_n=U, \, U_{n+1}=CUC^{-1}$ as
$
 \hat{\mathcal L}'=[\hat U, \hat {\mathcal L}],
 $
\begin{equation}\label{UFf1}
 \hat U=\begin{pmatrix} 0 & I\\ U-\lambda I & 0\end{pmatrix},\qquad
 \hat {\mathcal L}=\begin{pmatrix} C^{-1}F & C^{-1}\\ C^{-1}(F^2+B-\lambda I) & C^{-1}F\end{pmatrix}.
\end{equation}
The Lax equation can be written as the commutativity equation
\begin{equation}\label{dub1}
 [D_x-\hat U, \hat {\mathcal L}]=0,
\end{equation}
which is the matrix equation of the type studied by Dubrovin \cite{Dubrovin}.

Following Dubrovin, we call a matrix differential operator {\it finite-gap} if its vector-eigenfunction is meromorphic on a Riemann surface of finite genus.

By the general results of Dubrovin (see Lemma 4 in \cite{Dubrovin}) the relation (\ref{dub1}) implies that the
operator $D_x-\hat U$ is finite-gap with the spectral curve $\mathcal C$ given by $\det (\hat{\mathcal L}-\mu I)=0,$ or equivalently, by 
\begin{equation}\label{specC}
\det(\mathcal L+\lambda I)=0.
\end{equation}
The genus of the curve $\mathcal C$ generically is $g=(d-1)^2.$

\begin{theorem}
The matrix dressing chain (\ref{fbtop}) with $\alpha=0$  and generic $C$ is integrable in terms of $\theta$-functions of the Jacobi variety $J(\mathcal C).$ The corresponding Schr\"odinger operator $L=-D^2_x+U(x)$ is finite-gap with the spectral curve $\mathcal C.$
\end{theorem}

Indeed, $(D_x-\hat U)\Psi=0$ with $\Psi=(\psi_1, \psi_2)^T$ is equivalent to $L\psi_1=\lambda \psi_1,$ the details of the integration and formulae for the eigenfunctions can be found in \cite{Dubrovin} (see also more details in $2\times 2$ case in the next section).

Now we are ready to state the relation with the spinning tops.

\begin{theorem}
The matrix dressing chain (\ref{fbtop}) with $\alpha=0$ admits the real reduction with
\begin{equation}\label{redtop}
 F=-F^\top,\quad C=C^\top,\quad B=B^\top,
\end{equation}
which describes the rotation of $d$-dimensional rigid body about its centre of mass in the external field with an arbitrary quadratic potential. The motion of the body frame on the orthogonal group $SO(d)$ can be expressed explicitly in terms of the eigenfunctions of the corresponding finite-gap Schr\"odinger operator by formulas (\ref{phiB}),(\ref{phiU2}).
\end{theorem}

Indeed, the substitution $F=\Omega$, $J=C$, $B=P$ with $M=J\Omega+\Omega J=CF+FC$ into the matrix system (\ref{fbtop}) coincides with the equations of the motion of the top in quadratic potential (\ref{Bog2}). The Lax pair (\ref{lax}) agrees with the Lax pair (\ref{laxbog}) of Bogoyavlenskij and Reyman. To find the motion of the body frame $g(t) \in SO(d)$ we need to solve the equation $\dot g=\Omega g,$ which corresponds to the relation $\phi'=F \phi$ in (\ref{phiB}). However, according to (\ref{phiU2}), the solution can be found explicitly from the eigenfunctions of the corresponding finite-gap operator $L=-D^2+U$ (cf. \cite{Bogo92}). 

From the point of view of the spectral theory of matrix Schr\"odinger operator, the self-adjoint reduction $U=U^\top$ is more natural. In general, Darboux transformation does not preserve this reduction, and it is rather difficult to characterize the solutions of the dressing chain which lead to symmetric potentials. One class of solutions corresponds to equation (\ref{fbtop}) under the real reduction \footnote{Geometry and integrability of a different system on symmetric matrices is discussed in details in \cite{BBIMT}.}
\begin{equation}\label{symf}
 B=0,\quad F=F^\top,\quad CC^\top=I,
\end{equation}
or, the Hermitian reduction
\begin{equation}\label{symherm}
 B=0,\quad F=F^*,\quad CC^*=I
\end{equation}
with $X^*=\bar X^\top.$

However, there are other less obvious possibilities. Let us consider the chain (\ref{vbx}) instead of (\ref{fbx}). Elimination of $B$ yields
\[
 V''_{n+1}+V''_n=2V'_{n+1}(V_{n+1}-V_n)-2(V_{n+1}-V_n)V'_n
\]
and the periodicity condition $V_{n+1}=CV_nC^{-1}$ (we assume $\alpha=0$ here) brings to equation
\begin{equation}\label{vxx}
 V''C^{-1}+C^{-1}V''=2[V',[V,C^{-1}]].
\end{equation}
The solutions to this equations give rise to a special class of solutions of (\ref{fbtop}) with $\alpha=0$
given by the formulas
\begin{equation}\label{FB}
 F=CVC^{-1}-V,\,\,  B=CV'C^{-1}+V'-F^2, \quad U=2V'.
\end{equation}
The equation (\ref{vxx}) admits the reduction
\begin{equation}\label{symv}
 V=V^\top,\quad C=C^\top
\end{equation}
leading to the symmetric potential $U=2V'$. In general, this subclass of potentials is different from the one defined by (\ref{symf}). The reduction $V=V^\top$, $C=-C^\top$ is also admissible, however in this case the left hand side of equation (\ref{vxx}) becomes degenerate.

To end this section, we notice that (\ref{vxx}) is the Euler-Lagrange equation for the Lagrangian
\[
 L=\tr C^{-1}((V')^2+2VV'V).
\]
Another Lagrangian structure corresponds to the variable $\phi$ introduced by equation (\ref{phiB}). Under this change, the equation (\ref{fbtop}) at $\alpha=0$ takes the form
\[
 C\phi''\phi^{-1}+\phi''\phi^{-1}C=2(\phi'\phi^{-1})^2C-[C,\phi\Lambda\phi^{-1}]
\]
which is Euler-Lagrange equation for the Lagrangian 
\[
 L=\tr C((\phi'\phi^{-1})^2+\phi\Lambda\phi^{-1})
\]
in agreement with (\ref{Lagran}) and \cite{Bogo92}.
\section{Solutions of 1-periodic dressing chain in \texorpdfstring{$2\times2$}{2x2} matrix case}

In this section we consider the solutions of equation (\ref{fbtop}) for $2\times2$ matrix case.

\begin{theorem}
The general solution of periodic dressing chain (\ref{fbtop}) with $\alpha\neq 0$ for $2\times2$ matrices  can be expressed in terms of the Painlev\'e II and IV transcendents.

When $\alpha=0$ the general solution can be expressed in elliptic functions.
\end{theorem}

\begin{proof}
 Let
\[
 F=\begin{pmatrix} f_1 & f_2\\ f_3 & f_4 \end{pmatrix}, \quad
 B=\begin{pmatrix} b_1 & b_2\\ b_3 & -b_1\end{pmatrix}.
\]

First note that when $C=cI$ is a scalar matrix with $c\neq 0$, the equation (\ref{fbtop}) becomes 
$F'=\alpha I,$ so $F=\alpha x I + F_0.$
This leads to the reducible potentials with
$$
U= \begin{pmatrix} \alpha^2 (x-a_1)^2+b_1 & 0\\ 0 & \alpha^2 (x-a_2)^2+b_2 \end{pmatrix},
$$
so the problem is reduced to the standard case of the harmonic oscillator.

We will assume therefore, without loss of generality, that matrix $C$ has one of two forms:
\[
 \text{(a)}\quad
 C=\begin{pmatrix}1+\gamma&0\\ 0&1-\gamma\end{pmatrix},~~ \gamma\ne0,\pm1
 \qquad\text{or}\qquad
 \text{(b)}\quad
 C=\begin{pmatrix}1&1\\0&1\end{pmatrix}.
\]
In the case (a) the relations $\tr CF'=\alpha\tr C$, $\tr F'=2\alpha$ imply $f_1+f_4=2(\alpha x+c_1)$, $f_1-f_4=2c_2$ where $c_1$ and $c_2$ are integration constants and we come to the system
\begin{gather*}
 f'_2=2\gamma(\alpha x+c_1)f_2+\gamma b_2,\qquad
 f'_3=-2\gamma(\alpha x+c_1)f_3-\gamma b_3,\\
 b'_1=f_3b_2-f_2b_3,\qquad
 b'_2=2f_2b_1-2c_2b_2,\qquad
 b'_3=-2f_3b_1+2c_2b_3.
\end{gather*}
The substitutions
\[
 h=f_2f_3,\quad g_2=b_2/f_2,\quad g_3=b_3/f_3
\]
bring this to a subsystem for $g_2,g_3,b_1,h$:
\begin{equation}\label{sys1}
 g'_2= -2Xg_2-\gamma g^2_2+2b_1,\quad 
 g'_3= 2Xg_3+\gamma g^2_3-2b_1,\quad 
 b'_1=(g_2-g_3)h,\quad h'=\gamma(g_2-g_3)h
\end{equation}
where $X=\alpha\gamma x+\gamma c_1+c_2$, plus one quadrature for the function $f_2$. For $\alpha\ne0$, system (\ref*{sys1}) is equivalent to the Painlev\'e IV equation
\[
 y''=\frac{(y')^2}{2y}+\frac{3}{2}y^3+4zy^2+2\Bigl(z^2+\frac{c_3}{\alpha\gamma}-1\Bigr)y
 -\frac{2c_4}{\alpha^2y}
\]
for $y(z)=(\gamma/\alpha)^{1/2}g_2(x)$ and $z=(\alpha\gamma)^{-1/2}X$, with parameters defined by the values of the first integrals
\[
 c_3=h-\gamma b_1,\quad c_4=b^2_1+g_2g_3h.
\]

For $\alpha\ne0$, these first integrals bring (\ref*{sys1}) to an equation for $g_2(x)$ of the form $2g_2g''_2-(g'_2)^2=k_4g^4_2+k_3g^3_2+k_2g^2_2+k_0$ where $k_i$ are constants expressed through $\gamma$ and $c_i$ (this can be viewed as an autonomous version of P-IV). Note the absence of the first power of $g_2$ in the right hand side. Thanks to this, we have  $2g'''_2=(4k_4g^2_2+3k_3g_2+2k_2)g'_2$, by differentiating and canceling $g_2$, and this is easily integrated to $(g'_2)^2=P(g_2)$ where $P$ is a fourth degree polynomial, which means that $g_2$ is an elliptic function or its degeneration.

In the case (b) we find $f_1=X+g(x)$, $f_4=X-g(x)$ and $f_3=c_2$ where $X=\alpha x+c_1$, $c_1$ and $c_2$ are constants, then equation (\ref{fbtop}) is equivalent to the system
\begin{equation}\label{sys2}
\begin{gathered}
 f'_2=-2Xg-b_1,\qquad g'=c_2X+\frac{b_3}{2},\\
 b'_1=c_2b_2-f_2b_3,\qquad b'_2=2f_2b_1-2gb_2,\qquad b'_3=2gb_3-2c_2b_1.
\end{gathered}
\end{equation}
The first integrals
\[
 c_3=b_2b_3+b^2_1,\quad c_4=2c_2f_2-b_3+2g^2,\quad c_5=c_2X^2+c_2b_2+(X+f_2)b_3+2(b_1-\alpha)g
\]
allow us to reduce (\ref*{sys2}) to Painlev\'e XXXIV equation (which is a version of Painlev\'e II equation, see \cite[p.\,340]{Ince})
\begin{equation}\label{Pain}
 y''=\frac{(y')^2-1}{2y}+2\frac{\sqrt{c_3}}{\alpha}y^2-zy
\end{equation}
for $y(z)=\alpha^{1/3}(2c_2)^{-2/3}c^{-1/2}_3b_3(x)$ and $z=-(2\alpha c_2)^{-2/3}(2c_2X+c_4)$. 
In order to derive this equation we use only first integrals $c_3$ and $c_4$, solving them and the last equation of (\ref{sys2}) with respect to $b_1$, $b_2$ and $f_2$. Substituting this into the equation for $b''_3$ obtained from (\ref{sys2}) we find that the variable $g$ cancels identically and the equation for $b_3$ can be reduced to the form (\ref{Pain}) by the above linear transformations. The remaining first integral $c_5$ is needed only in order to express $g$ as a rational function of $b_3$, $b'_3$ and $X$; after this $b_1$, $b_2$ and $f_2$ are also expressed in this way.

When $\alpha=0$ this procedure leads to an autonomous equation of the form $2b_3b''_3-(b'_3)^2=k_3b^3_3+k_2b^2_3+k_0$ of the same type as in the previous case, and its integration brings to an equation $(b'_3)^2=P(b_3)$ with a third degree polynomial $P$. (More details on the case $\alpha=0$ are given in the examples 2 and 3 below.)

\end{proof}

\section{Spectral properties of the matrix Schr\"odinger operators}

The spectral theory of the Schr\"odinger operators 
\begin{equation}\label{schper}
L=-D_x^2+U(x), \,\, U(x)=U(x+T)=U^\top(x)
\end{equation}
 with periodic symmetric matrix potential $U(x)$ goes back to the work of Lyapunov \cite{Lyapunov} and Krein \cite{Krein} in stability theory.

The spectrum of a matrix Schr\"odinger operator can be defined as the set of $\lambda \in \mathbb C$ such that the corresponding Schr\"odinger equation $L\psi=\lambda \psi$ has a bounded solution \cite{Kuchment}. For the operators (\ref{schper})  the spectrum is purely absolutely continuous and consists of (in general, infinitely many) bands with different even multiplicities from 0 (spectral gaps) up to maximal $2d$.

We call the Schr\"odinger operator (\ref{schper}) {\it maximally finite-gap} if for all sufficiently large positive $\lambda$ all solutions of the corresponding Schr\"odinger equation $L\psi=\lambda\psi$ are bounded (cf. \cite{GS}). 
In the reducible case with $U(x)=diag (u_1(x), \dots, u_d(x))$ this happens if and only if all $u_1(x), \dots, u_d(x)$ are the finite-gap potentials in the usual sense \cite{Nov}.

We also call the operator (\ref{schper}) {\it maximally $n$-band} if its absolutely continuous spectrum consists of $n$ spectral bands (including infinite one) with maximal multiplicity $2d.$
As in the usual case, we will extend this notion to the case of quasi-periodic matrix potentials.

In the periodic case we have the Bloch-Floquet multipliers $\tau_i(\lambda)$, $i=1,\dots,2d$ defined as the eigenvalues of the corresponding monodromy matrix $M(\lambda)$ on the space of solutions $L\psi=\lambda \psi$, satisfying characteristic equation
$$
\det (M(\lambda)-\tau I)=0.
$$
It is known that if $\tau(\lambda)$ is a multiplier, then $\tau^{-1}(\lambda)$ is a multiplier too. For real $\lambda$ the same is true for $\overline{\tau(\lambda)}.$ The spectrum of $L$ has the band structure with the intervals determined by the condition $|\tau_i(\lambda)|=1$. 
The ends of the bands correspond to the collision of multipliers which
may happen either at $\pm 1$ (for periodic or antiperiodic
eigenfunctions), or somewhere on the unit circle (for the so-called {\it
resonances}), see more details in \cite{BBK, ChK}.


Let $M(t), P(t)$ be a solution (in general, quasi-periodic) of the Bogoyavlenskij system of the rigid body motion in Newtonian field with quadratic potential
  \begin{equation}\label{motion}
   \dot M=[M,\Omega]-[P, J], \quad \dot P=[P, \Omega],
\end{equation}   
where $M=J\Omega+\Omega J, \,\, \Omega^\top=-\Omega, \, \, P^\top=P$ and $J=J^\top$ is given symmetric matrix assumed to be diagonal: $J=diag (J_1,\dots, J_d), J_1\geq J_2\geq \dots \geq J_d>0.$ 

 Replace here $t$ by $x$ and consider the matrix Schr\"odinger operator
  \begin{equation}\label{schbog}
L=-D_x^2+U(x), \,\, U=\Omega^2-\Omega'+P.
\end{equation} 

\begin{theorem}
The matrix Schr\"odinger operator (\ref{schbog}) is self-adjoint maximally finite-gap operator with the spectral curve $\Gamma$ given by
\begin{equation}\label{spectop}
 \det(\mu^2J^2-\mu M-P+\lambda I)=0
\end{equation}
with $\bar\mu=-\mu, \, \bar\lambda =\lambda$, which is a real version of the classical spectral curve (\ref{laxspec}).
\end{theorem}

\begin{proof}

Although the corresponding potential $U=\Omega^2-\Omega'+P$ in general is not symmetric, it is in fact formally self-adjoint with respect to the scalar product defined by $J.$
Indeed, if we define the real symmetric matrix $S$ such that $S^2=J,$ then the corresponding potential $\tilde U=SUS^{-1}$ is symmetric:
$$
\tilde U^T-\tilde U=[S^{-1}(\Omega^2+\Omega'+P)S-S(\Omega^2-\Omega'+P)S^{-1}]=S^{-1}(M'-[J,\Omega^2+P])S=0$$
due to (\ref{motion}) and identity $[M, \Omega]=[J, \Omega^2].$ 

From the relation with the dressing chain (see Theorem 2) we have the following important corollary (see also formula (\ref{Nov2}) in the last section).

\begin{lemma}
The Schr\"odinger operator (\ref{schbog}) commutes with the operator $A=J^{-1}(D_x+\Omega)$.
\end{lemma}

Indeed, we know that the Darboux transformed operator $\tilde L=CLC^{-1}, C=J$
satisfies the relation (\ref{factorization}) $$AL-\tilde LA=0, \, A=D_x+F, F=\Omega,$$ so we have
$AL-CLC^{-1}A=C[C^{-1}A, L]=0$ and thus $[L, C^{-1}A]=0.$

The joint eigenfunction $\psi$ of $L$ and $C^{-1}A$ satisfies $L\psi =\lambda \psi, \,\, C^{-1}A \psi =\mu \psi$, and thus
$\psi'=(\mu J-\Omega)\psi.$
Since $J=J^\top$ and $\Omega^\top=-\Omega$ we see that
\begin{equation}\label{egfnorm}
\frac{d}{dx}|\psi|^2=(\mu+\bar\mu) \psi^\dagger J \psi=(\mu+\bar\mu)\sum_{k=1}^d J_k |\psi_k|^2.
\end{equation}
Since $\sum_{k=1}^d J_k |\psi_k|^2\geq J_d |\psi |^2$ this implies that the corresponding eigenfunction is bounded on the whole $x$-axis if and only if $\mu+\bar\mu=0$, which leads to the real version of the classical spectral curve (\ref{spectop}) with $\bar\mu=-\mu, \, \bar \lambda =\lambda.$

The spectrum of $L$ is the projection of $\Gamma$ on the $\lambda$-axis and consists of the spectral bands with different multiplicities given by the number of pre-images, which is the number of purely imaginary roots of equation (\ref{spectop}) with given real $\lambda.$ Since for large positive $\lambda$ we have $2d$ purely imaginary roots $\mu_k \sim \pm i \sqrt{\lambda} J_k^{-1}$, the operator $L$ is maximally finite-gap. More precisely, if $\lambda_1\geq \lambda_2\geq \dots \geq \lambda_d$ be the eigenvalues of $P$, then all $\lambda>\lambda_1$ belong to the continuous spectrum of $L$ with maximal multiplicity $2d$ (see Fig. 1 and 2 below). 
\end{proof}

 We will discuss now the  $2\times 2$ case of the matrix dressing chain (\ref{fbtop}) in more detail.

\section{Explicit examples in $2\times 2$-matrix case}

\subsection*{Example 1. Two-dimensional top in quadratic potential.}~

\medskip
In two dimensions one can think of a non-symmetric plate rotating/librating in a plane about its centre of mass in the Newtonian field with quadratic potential. 

As we have seen, the equation of motion (\ref{Bog2}) is the reduction (\ref{redtop}) of the dressing chain (\ref{fbtop}) with $\alpha=0$, which in $2\times 2$-case gives
\[
 F=\Omega=\begin{pmatrix} 0 & \omega \\- \omega & 0 \end{pmatrix}, \quad
 B=P=\begin{pmatrix} u & v\\ v & -u\end{pmatrix}, \quad
 C=J=\begin{pmatrix} J_1&0\\ 0& J_2\end{pmatrix},
\]
where we use the shift $B \to B+cI$ to assume that the trace of $B$ is zero.
This leads to the system
\begin{equation}\label{2dtop}
\dot u=-2\omega v, \quad \dot v=2\omega u, \quad \dot \omega=\beta v,
\end{equation}
where $\beta=(J_1-J_2)/(J_1+J_2)$.
 This is the Hamiltonian system with Poisson bracket 
 \begin{equation}\label{PB}
 \{\omega,u\}=v, \quad  \{\omega,v\}=-uv, \quad  \{u,v\}=0,
 \end{equation}
having Casimir function $I=u^2+v^2:=R^2$ and the Hamiltonian $H=\omega^2+\beta u$, so the system is equivalent to the mathematical pendulum with
$$
\mathcal H=p^2+\beta R \cos q, \quad p=\omega, \, u=R\cos q, \,  v=R \sin q.
$$ 
We have \begin{equation}\label{Jac}
\dot\omega^2=\beta^2 v^2=\beta^2 (R^2-u^2)=\beta^2 R^2 - \beta^2 u^2 =\beta^2 R^2-(H_0-\omega^2)^2,
\end{equation}
where $H_0=\omega_0^2+\beta u_0>-|\beta|R$ is determined by the initial data. 
This implies that 
\begin{equation}\label{omega}
\omega(t)= A\dn(At-t_0, k),\quad A^2=|\beta|R+H_0,\quad k^2=\frac{2|\beta|R}{|\beta|R+H_0},
\end{equation} 
where $\dn(z,k)$ is the classical Jacobi elliptic function, satisfying the differential equation $(\dn')^2=(1-\dn^2)(\dn^2-k'^2), \,\, k'^2=1-k^2$ (see e.g. \cite{WW}).
\medskip

The corresponding Schr\"odinger operator has matrix potential 
\begin{equation}\label{Utop}
U=\begin{pmatrix}u-\omega^2 & v-\dot \omega\\ v+\dot \omega & -u-\omega^2 \end{pmatrix}= \begin{pmatrix}(1+\beta)u& (1-\beta)v\\ (1+\beta)v& -(1-\beta)u \end{pmatrix}-H_0 I.
\end{equation} 
As we have already mentioned after conjugation by $S=\sqrt{J}$ the potential becomes symmetric:
\begin{equation}\label{Utoptil}
\tilde U=S U S^{-1}= \begin{pmatrix}(1+\beta)u& \gamma v\\ \gamma v& -(1-\beta)u \end{pmatrix}-H_0 I, \quad \gamma=\sqrt{1-\beta^2}=\frac{2\sqrt{J_1J_2}}{J_1+J_2}.
\end{equation} 
\begin{prop}
The matrix Schr\"odinger operators with potentials (\ref{Utop}), (\ref{Utoptil}) are maximally $n$-band with $n\leq 2$ and the eigenfunctions expressible in elliptic functions.
\end{prop}

\begin{proof}
Indeed, the spectral curve $\det (\mathcal L+\lambda I)=0$ with 
$$
\mathcal L=\mu^2C^2-\mu(CF+FC)-B= \begin{pmatrix} \mu^2J_1^2-u & -\mu(J_1+J_2)\omega-v\\ \mu(J_1+J_2)\omega-v & \mu^2 J_2^2+u \end{pmatrix}
$$
has the form
 $$
 \mu^4J_1^2J_2^2+\mu^2(J_1+J_2)^2[H_0+\frac{J_1^2+J_2^2}{(J_1+J_2)^2}\lambda]+\lambda^2-R^2=0.
 $$
Introducing $\rho:=i\mu \sqrt{J_1J_2}$ and $\delta=\frac{J_2}{J_1},$ we have the equation of the spectral curve
\begin{equation}\label{Spectop}
\rho^4-\rho^2\left[h+(\delta^{-1}+\delta)\lambda\right]+\lambda^2-R^2=0,
\end{equation} 
where $h:=\frac{(1+\delta)^2}{\delta} H_0.$ The discriminant of this quadratic (in $z=\rho^2$) equation equals $4h^2-4(\delta^{-1}-\delta)^2R^2$, so
$
h_\pm= \pm(\delta^{-1}-\delta)R
$
are the critical values, corresponding to the critical values of the energy
$$
H_\pm= \pm\frac{J_1-J_2}{J_1+J_2} R =\pm\beta R.
$$
The values $h>h_+$ correspond to the rotations of pendulum, while for values $h_{-}<h<h_+$ we have the librations.

Note that for the energy $H_0\neq H_+$ the solutions of the system (\ref{2dtop}) are periodic, so the corresponding Schr\"odinger operator 
$$
L=-D_x^2+\begin{pmatrix}(1+\beta)u& \gamma v\\ \gamma v& -(1-\beta)u \end{pmatrix}-H_0 I
$$
has real symmetric periodic matrix potential.

The analysis of the spectral curves (\ref{Spectop}) (which are real parts of elliptic curves) shows that we have maximally 1-band and 2-band cases respectively, depending on whether $h_- \leq h<  h_*$ (case A) and $h>h_*$ (case B), where
$
h_*:=(\delta+\delta^{-1})R.
$

In case A $(-\beta R\leq H_0<\frac{1+\beta^2}{2}R)$ the spectral bands are 
 $-R\leq\lambda\leq R\, (\operatorname{multiplicity} 2)$ and $\lambda\geq R \,(\operatorname{multiplicity} 4),$
in case B $(H_0>\frac{1+\beta^2}{2}R)$ we have 
$ -R\leq \lambda \leq R\, (\operatorname{multiplicity} 2)$ and $ \lambda_0\leq \lambda\leq -R, \,\,\, \lambda \geq R \, (\operatorname{multiplicity} 4),$
where $\lambda_0$ is the largest root of the quadratic equation
$$
(\delta-\delta^{-1})^2\lambda^2+2h(\delta+\delta^{-1})\lambda+h^2+4R^2=0.
$$
For example, when $\delta=\frac{1}{2}, \, R=10$, we have $h_\pm=\pm \frac{3}{2}R=\pm 15, \,\,\,\,\, h_*=\frac{5}{2}R=25$
and two types of the spectral curves shown on Fig. \ref{fig:spec2dtop}. The corresponding spectral bands with multiplicities 4 and 2 are indicated on $\lambda$-axis in black and grey respectively.

Note that in the second (maximally 2-band) case the end of the first spectral band is not a multiplicity 1 periodic/anti-periodic level, but a multiplicity 2 resonance level, which is a pure matrix phenomenon (see more on this in \cite{BBK,ChK}). 
\end{proof}

\begin{figure}[!ht]
\includegraphics[width=0.35\textwidth]{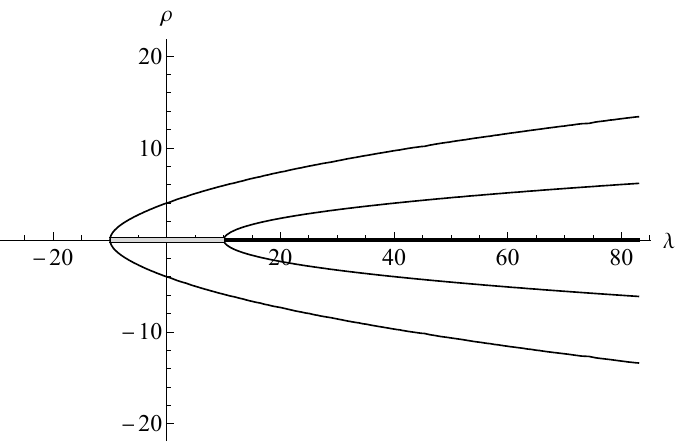}  \quad \includegraphics[width=0.35\textwidth]{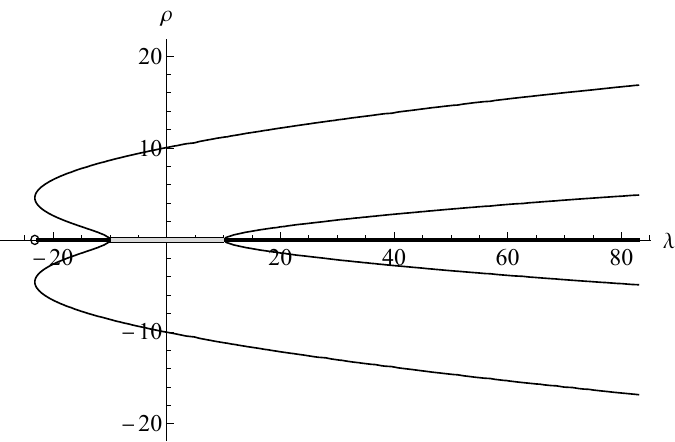}
\caption{Spectral curves (\ref{Spectop}) of with $\delta=1/2$, $R=10$ and $h=10$ (left) and $h=100$ (right).}\label{fig:spec2dtop}
\end{figure}

Let us comment briefly on the 3D case with
\[
 F=\begin{pmatrix}
  0 & \omega_1 & \omega_2 \\
  -\omega_1 & 0 & \omega_3 \\
  -\omega_2 & -\omega_3 & 0 
 \end{pmatrix},\quad
 B=\begin{pmatrix}
  u_1 & v_1 & v_2 \\
  v_1 & u_2 & v_3 \\
  v_2 & v_3 & -u_1-u_2 
 \end{pmatrix},\quad
 C=\diag(J_1,J_2,J_3). 
\]

In that case the Hamiltonian system is integrable with Liouville tori of dimension 3 and the general solution is expressed in terms of Prym $\theta$-functions \cite{Bogo92}. 
On Fig. 5 we show different real versions of the spectral curves (of genus 4) with 
the following values of $(J_1,J_2,J_3;u_1,u_2;v_1,v_2,v_3;\omega_1,\omega_2,\omega_3)$ (in the order of plots in the figure):
\[
\begin{array}{ll}
 (1,2,3; 8,0; 1,2,1; 0,1,1), & (1,2,3; 1,2,10; 5,2,2),\\
 (1,2,3; 1,2; 1,2,1; 1,2,2), & (1,2,10; 1,2; 1,2,30; 5,2,2),
\end{array}
\]
where the spectral bands with maximal multiplicity 6 are marked in black. We can see maximally $n$-band examples with $n\leq 3$, but whether $n=3$ is maximal still to be shown.

\begin{figure}[!ht]
\centerline{\includegraphics[width=0.55\textwidth]{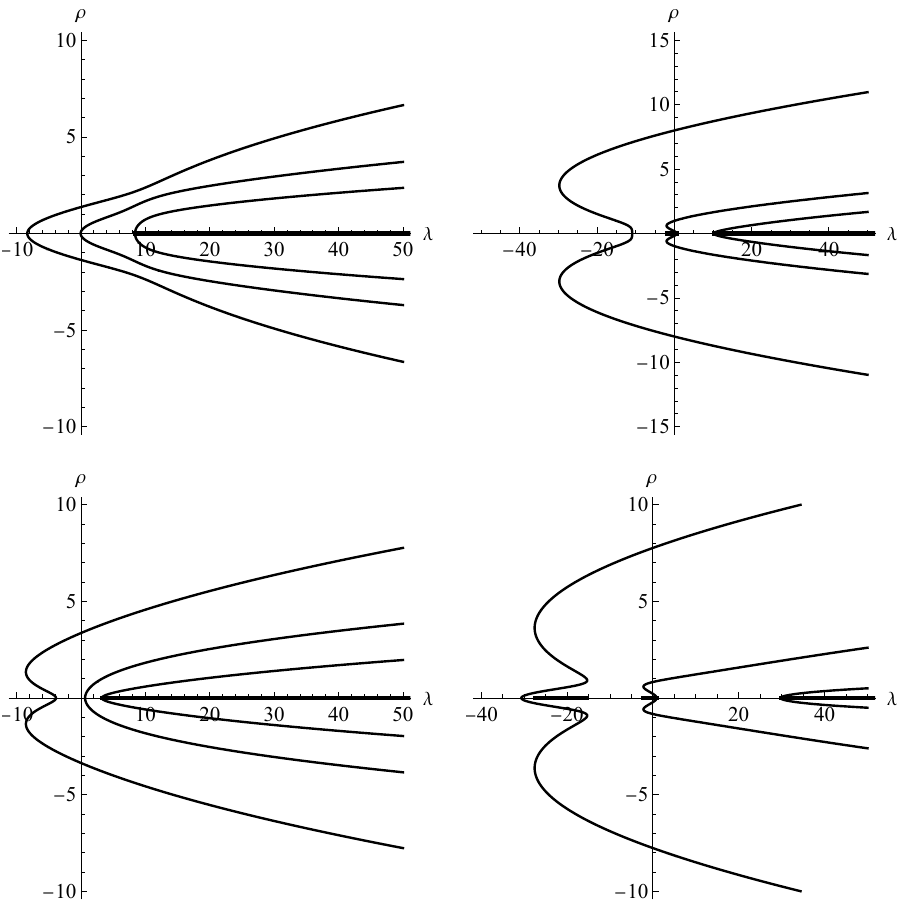}}
\caption{Spectral curves for Bogoyavlenskij top in $3\times3$ case.}\label{fig:spec3dtop}
\end{figure}


\subsection*{Example 2. Self-adjoint reduction (\ref{symf}).}~ 
\medskip

We will consider here only real reduction (\ref{symf}) since the Hermitian reduction (\ref{symherm}) in $2\times 2$ does not lead to the principally new cases.

Let $B=0$ and
\begin{equation}\label{sym2}
 F=\begin{pmatrix} f & g \\ g & h \end{pmatrix},\qquad
 C=\left(\begin{array}{rr}
   \cos\omega_0 & -\sin\omega_0 \\
   \sin\omega_0 & \cos\omega_0
  \end{array}\right),\quad \varkappa=\tan\omega_0.
\end{equation}
Then equation (\ref{fbtop}) is equivalent to the system
\begin{equation}\label{sys3}
 f'=\alpha-\varkappa(f+h)g,\quad h'=\alpha+\varkappa(f+h)g,\quad 2g'=\varkappa(f^2-h^2),
\end{equation}
 which is easy to solve in
elementary functions using the first integrals $I_1=f+h-2\alpha x$
and $I_2=(f-h)^2+4g^2$ (see the explicit formulas below). These simplest examples allow us to demonstrate the main characteristic features of the corresponding operators in both cases $\alpha=0$ and $\alpha\ne0$.\smallskip

\medskip
\noindent{\em Case $\alpha=0$: pseudo-Mathieu matrix potentials.} 

\medskip

We find
\[
 F=k_1I+k_2\CSF,\quad \chi_1=2\varkappa k_1x+\omega_1,
\]
where $k_1,k_2,\omega_1$ are integration constants. The potential 
\begin{equation}\label{e1:u}
 U=F^2-F'=k^2I+2\beta\CS,\quad \chi=2\gamma x+\delta,
\end{equation}
is periodic with
$k^2=k_1^2+k_2^2, \,
 \beta=\frac{k_1k_2}{\cos\omega_0},\,
 \gamma=\varkappa k_1,\, \delta=\omega_1-\omega_0.
$

Without loss of generality, we can consider the family of matrix Schr\"odinger operators with $\pi$-periodic potentials
\begin{equation}\label{A}
 U=A\CSX,
\end{equation}
which remind the Mathieu operators $-D_x^2+A \cos 2x$, but, as we will see, have very different spectral properties.
In particular, we prove the following result.

\begin{prop}
The matrix Schr\"odinger operator with potential  (\ref{A}) is a maximally finite-gap operator with the eigenfunctions explicitly expressed in elementary functions.
\end{prop}

\begin{proof}
Indeed, the corresponding eigenfunction can be found explicitly as follows. 

Changing $\psi \to \phi=R \psi, 
R=\left(\begin{array}{rr}
   \cos x & \sin x \\
   -\sin x & \cos x
  \end{array}\right),$ we reduce the Schr\"odinger equation
$$
-\psi_{xx}+A\CSX\psi=\lambda \psi
$$
to the following equation with constant matrix coefficients
\begin{equation}\label{mat}
-\phi_{xx}-2\Omega \phi_x+AJ\phi=(\lambda -1)\phi, \quad J=\begin{pmatrix} 1 & 0 \\ 0 & -1 \end{pmatrix}, \,\,  \Omega=\begin{pmatrix} 0 & -1 \\ 1 & 0 \end{pmatrix}.
\end{equation}

Substituting $\phi=e^{i\rho x}\phi_0$ one obtains the biquadratic characteristic equation
\begin{equation}
\label{chareq}
 \rho^4-2(\lambda+1)\rho^2+(\lambda-1)^2-A^2=0
\end{equation}
and the eigenvector
$$
 \phi_0=\binom{2i\rho}{\lambda-A-1-\rho^2}.
$$
Since a polynomial $\rho^4-a\rho^2+b$ with real coefficients has two real zeroes if and only if $b<0$ and four real zeroes if and only if $a>0$, $b>0$, $a^2>4b$, the spectrum of our Schr\"odinger operator is defined by the following inequalities:
\begin{align*}
 &\operatorname{multiplicity}=2:\quad (\lambda-1)^2-A^2<0, \\
 &\operatorname{multiplicity}=4:\quad \lambda>-1,\quad (\lambda-1)^2-A^2>0,\quad \lambda>-A^2/4.
\end{align*}
We assume $A>0$ without loss of generality. The order of the points $-1$, $-A^2/4$, $1-A$, $1+A$ on the real axis is defined by the sign of $A-2$, as shown schematically on fig.~\ref{fig:rholambda}, where the spectral bands of multiplicity 2 and 4 are marked in grey and black respectively.

There are infinitely many gaps, which are all closed and not shown on the scheme. For example, for the critical amplitude $A=2$ there are two closed gaps in the first spectral band, corresponding to $\rho=1$ and $\rho=2$ with $\lambda=2-2\sqrt{2}$ (periodic) and $\lambda=5-2\sqrt{5}$ (anti-periodic) levels respectively.
\end{proof}

\begin{figure}[!ht]
\centerline{\includegraphics[width=0.5\textwidth]{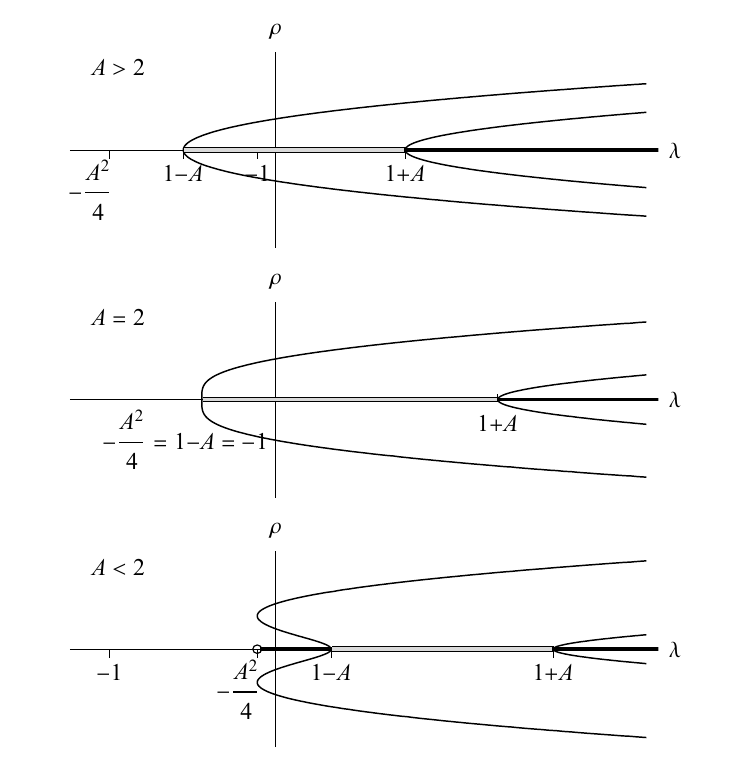}}
\caption{Characteristic curves (\ref{chareq}) and spectral bands}\label{fig:rholambda}
\end{figure}


\begin{remark}
Note that in this case the characteristic curve (\ref{chareq}) determining the spectrum of the Schr\"odinger operator and the spectral curve of the dressing chain given by $$\det(\mu C^2-(CF+FC)-\nu I)=0$$ 
are related by a non-trivial change
$$
\rho=i(a\mu+\sin\omega_0 k_1), \quad \lambda=\nu-a^2\mu^2- b\mu +c, \quad A=2k_2\cos\omega_0
$$
with
$
a=-\sin\omega_0\cos\omega_0, \,\, b=\cos 2\omega_0-2k_1\sin^2\omega_0\cos\omega_0, \,\, c=1-k_1^2\sin^2\omega_0 +2k_1\cos\omega_0.
$
\end{remark}

\medskip
\noindent{\em Case $\alpha\ne0$: exotic matrix harmonic oscillators.} 
\medskip

Assuming, without loss of generality, $f+h=2\alpha x$, $\alpha>0$, we obtain the solution of the system (\ref{sys3}) and the corresponding potential in the form (see fig.~\ref{fig:oscillator_U})
\begin{gather}
\nonumber
 F=\alpha xI+k\CSF,\quad \chi_1 =\gamma x^2+\omega_1,\quad \gamma=\alpha\tan(\omega_0),\\
\label{pot2}
 U=(\alpha^2x^2+k^2-\alpha)I +\frac{2\alpha kx}{\cos\omega_0}\CS,\quad \chi=\gamma x^2+\omega_1-\omega_0.
\end{gather}
When $k=0$ we have the scalar harmonic oscillator $L=-D_x^2+(\alpha^2x^2-\alpha)I$
with the discrete spectrum 
\begin{equation}
\label{disc}
\lambda=2n\alpha, \, n \in \mathbb Z_{\geq 0}
\end{equation}
 with all levels of multiplicity 2.
We will show now that for $k\neq 0$ the spectrum remains the same. 

\begin{prop}
The matrix versions of the harmonic oscillators with potential (\ref{pot2}) have the discrete spectrum (\ref{disc}) with the eigenfunctions (in case of non-zero $k$ and $\gamma$) expressed in terms of the Weber (parabolic cylinder) functions.
\end{prop}

\begin{proof}

From (\ref{factorization}) we have the operator relations
\[
 L=\bar AA,\quad L+2\alpha=C^{-1}A\bar AC,\quad A=D_x+F,\,\, \bar A=-D_x+F,
\]
which show that $\bar AC$ is the raising operator: let $L\psi=\lambda\psi$ then the function $\tilde\psi=\bar AC\psi$ satisfies equation
\[
 L\tilde\psi=\bar AA\tilde\psi= \bar AA\bar AC^{-1}\psi= \bar AC(L+2\alpha)\psi
 =(\lambda+2\alpha)\tilde\psi.
\]
In order to find the ground state at $\lambda=0$, we solve the equation $A\psi=0$. Let $\psi=\binom{\psi_1}{\psi_2}$, then the functions $\phi=\psi_1+i\psi_2$, $\bar\phi=\psi_1-i\psi_2$ satisfy the system
\[
 \phi'=-\alpha x\phi-ke^{i(\gamma x^2+\omega_1)}\bar\phi,\quad
 \bar\phi'=-\alpha x\bar\phi-ke^{-i(\gamma x^2+\omega_1)}\phi
\]
and the substitution $\phi=e^{-\alpha x^2/2}\varphi$ brings to equation
\begin{equation}\label{phieq}
 \varphi''-2i\gamma x\varphi'-k^2\varphi=0.
\end{equation}
 When $\gamma=0$ we have general solution $\varphi=C_1 e^{kx}+C_2 e^{-kx},$ implying that for some $C>0$
\begin{equation}\label{ineq}
|\varphi(x)| < Ce^{|kx|}.
\end{equation}

We claim that the same estimate works also for $\gamma\neq 0.$ To show this rewrite the complex equation (\ref{phieq}) as the system 
$
w''-2\gamma x r'=0, \quad r''+2\gamma xw-4k^2 r^2=0
$
for two real variables
$r(x):=\frac{1}{2}(|\varphi(x)|^2+k^{-2}|\varphi'(x)|^2)$ and $ w(x):=i(\varphi(x)\bar \varphi'(x)-\varphi'(x)\bar \varphi(x).$
It is easy to check that this system has integral $I=(r')^2+w^2-4k^2r^2,$ which implies that 
$|r'|< 2|k|r+\const,$ and thus the inequality (\ref{ineq}). 

This means that the corresponding $\phi=e^{-\alpha x^2/2}\varphi$ and hence $\psi$ belong to $L_2({\mathbb R}),$ so we have a two-dimensional space of the eigenfunctions of the operator $L$ at the groundstate $\lambda=0$ . Applying the raising operator $\bar AC$ we get the eigenfunctions for $\lambda=2n\alpha$, $n\in \mathbb Z_{\geq 0}$. 

Note that when $\gamma\neq 0$ the equation (\ref{phieq}) can be reduced to the Hermite equation
$y_{zz}-2zy_z+2\nu y=0$ with $\nu=\frac{ik^2}{2\gamma}$ by the scaling $\varphi(x)=y(z)$, $z=\sqrt{i\gamma}x$. 
From the theory of Weber functions (see e.g. Chapter 16 in \cite{WW}) it follows that for non-zero $k$ and $\gamma$ all solutions $\varphi(x)$ of (\ref{phieq}) asymptotically decay as $|x|^{-1/2}$ and thus bounded for real $x$ (see fig.~\ref{fig:oscillator_phi}). This means that the eigenfunctions of the operators (\ref{pot2}) can be expressed in terms of Weber (parabolic cylinder) functions.
\end{proof}

\begin{figure}[!ht]
\centerline{\includegraphics[width=0.65\textwidth]{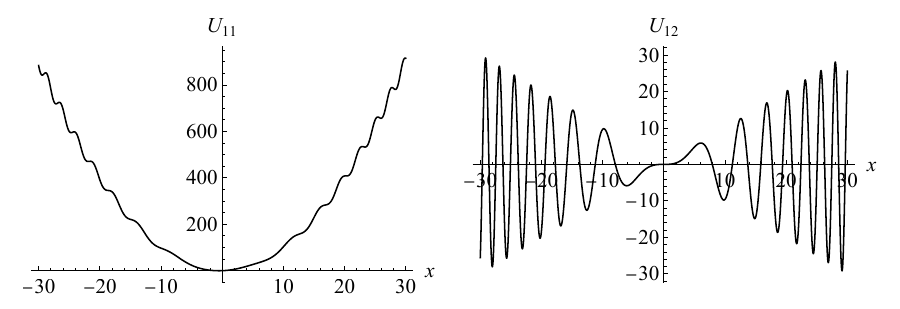}}
\caption{Components $U_{11}$ and $U_{12}$ of the matrix potential (\ref{pot2}), for $\alpha=1$, $k=0.5$, $\gamma=0.05$ and $\omega_1-\omega_0=0$.}\label{fig:oscillator_U}
\end{figure}

\begin{figure}[!ht]
\centerline{\includegraphics[width=0.65\textwidth]{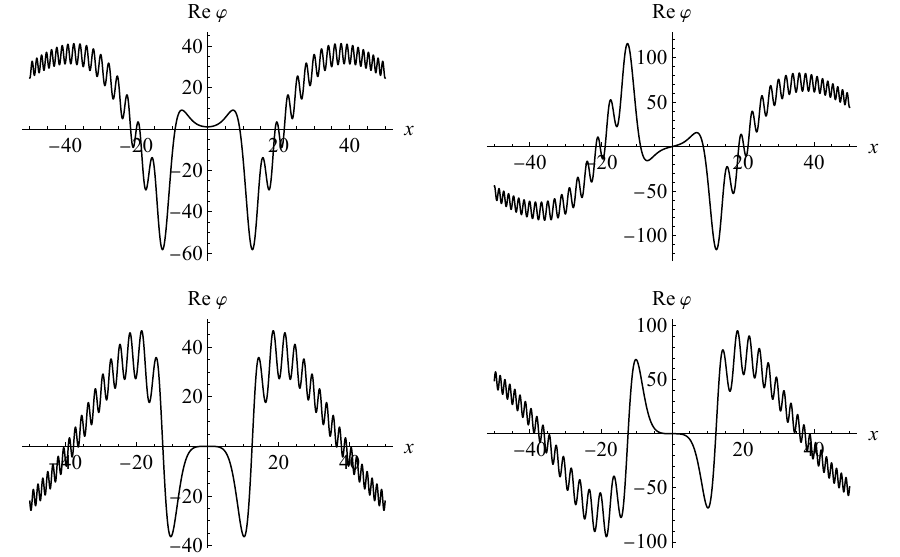}}
\caption{Real parts of solutions of equation (\ref{phieq}), for $k=0.5$ and $\gamma=0.05$, with initial data $(\varphi(0),\varphi'(0))$ equal to $(1,0)$ and $(0,1)$ (top), $(i,0)$ and $(0,i)$ (bottom).}\label{fig:oscillator_phi}
\end{figure}

\subsection*{Example 3. Non-trivial symmetry reduction (\ref{symv}).}~ 

\medskip

Consider now the solutions of (\ref{fbtop}) of the form (\ref{FB}),
where $V=V^\top$ satisfies equation (\ref{vxx})
with $C=C^\top.$
Let
\[
  V=\begin{pmatrix} f & g \\ g & h \end{pmatrix},\quad
  C^{-1}=\diag(1,\gamma),\quad \gamma\ne0,\pm1
\]
then equation (\ref{vxx}) turns into the system
\[
 f''=2(1-\gamma)gg',\quad h''=2(1-\gamma^{-1})gg',\quad
 g''=2\frac{1-\gamma}{1+\gamma}g(h'-f')
\]
which is solved in general in terms of Jacobi elliptic sinus $s=\sn(x,k)$, satisfying 
\[
 (s')^2=(1-s^2)(1-k^2s^2).
\] 
Up to the changes $x\to c_1x+c_2$, $\lambda\to\lambda+c_3$, the potential $U=2V'$ is of the form
\begin{equation}\label{e3:u}
 U=\frac{1}{1-\gamma}\begin{pmatrix}
  -2\gamma k^2s^2 & 2\sqrt{-\gamma}ks' \\
  2\sqrt{-\gamma}ks' & 2k^2s^2-(1+k^2)(1+\gamma)
 \end{pmatrix}.
\end{equation}
The $\psi$-function can be found in quadratures by solving the system $\psi'=-{\mathcal A}\psi$, ${\mathcal L}\psi=0$ with the matrices (\ref{PQfb}), where
\[
 F=CVC^{-1}-V,\quad B=CV'C^{-1}+V'-F^2.
\]
The condition $\det{(\mathcal L +\lambda I)}=0$ defines the spectral curve $R(\lambda,\mu)=0$ and substituting $\psi=\psi_1\dbinom{1}{-({\mathcal L}_{11}+\lambda)/{\mathcal L}_{12}}\Big|_{R(\lambda,\mu)=0}\in\ker{\mathcal L}$ into the equation $\psi_x=-{\mathcal A}\psi$ we find $\psi_1$ from the equation of the form $\psi'_1=-({\mathcal A}_{11}-{\mathcal A}_{12}({\mathcal L}_{11}+\lambda)/{\mathcal L}_{12})\psi_1$. For the potential (\ref{e3:u}), the spectral curve is
\begin{multline*}
 R(\lambda,\mu)= (1-\gamma)^2(\gamma^2\lambda^2+\mu^4)
  +(1+k^2)(1-\gamma^2)\gamma^2(\lambda+\mu^2)\\
  +(1-\gamma)(1-\gamma+\gamma^2-\gamma^3)\lambda\mu^2
  +k^2\gamma^2(1+\gamma)^2=0
\end{multline*}
and the equation for $\psi_1$ reads
\[
 \frac{\psi'_1}{\psi_1}
  = \frac{\gamma\mu s'+k^2\gamma s^3+\frac{1-\gamma}{1+\gamma}(\gamma\lambda-\mu^2)s}
         {\gamma s'+(\gamma-1)\mu s}\bigg|_{R(\lambda,\mu)=0}.
\]
This formula provides the fundamental system of solutions since we have, in general, 4 values of $\mu$ for the given $\lambda$.


In particular, at $k=1$ we obtain, after some transformations, the special soliton-like potential
\[
 U= \frac{1}{(1-\gamma)\cosh^2x}\begin{pmatrix}
  2\gamma & -2\sqrt{\gamma}\sinh x \\
  -2\sqrt{\gamma}\sinh x & (1+\gamma)\cosh^2x-2
  \end{pmatrix}
\]
with the reducible spectral curve
\[
  (\mu^2+\lambda)\Bigl(\mu^2+\gamma^2\Bigl(\lambda-\frac{1+\gamma}{1-\gamma}\Bigr)\Bigr)=0
\]
and $\psi$-functions
\begin{gather*}
 \psi=e^{\nu x}
 \begin{pmatrix}
   (\gamma-1)\nu-\gamma\tanh x\\
   -\sqrt{\gamma}\sech x
 \end{pmatrix},\quad \nu=\sqrt{-\lambda}, \\
 \psi=e^{\nu x}
 \begin{pmatrix} \sech x \\
  \dfrac{\gamma^2+(1-\gamma)^2\lambda-\sech^2x}
   {\sqrt{\gamma}((\gamma-1)\nu-\tanh x)}
 \end{pmatrix},\quad \nu=\sqrt{\frac{1+\gamma}{1-\gamma}-\lambda}.
\end{gather*}

We leave more detailed discussion of the spectral properties of the corresponding matrix Schr\"odinger operators for another occasion.

\section{Matrix KdV hierarchy and Novikov equation}

As in the scalar case, Darboux transformation can be prolonged to the B\"acklund transformation for the matrix KdV equation. The isospectral symmetries for Schr\"odinger equation (\ref{psixx}) are of the form
$
 \psi_t=P\psi+Q\psi_x
$
where $P$ and $Q$ satisfy relations
\[
 2P_x+Q_{xx}=[U,Q],\quad U_t=P_{xx}+QU_x+2Q_xU+[P,U]-2\lambda Q_x.
\]
This leads to the hierarchy
$
 U_{t_k}=-2Q_{k,x},
$
where $Q_k$ satisfy recurrent relation 
\[
 Q_0=\const,\quad -4Q_{k+1,x}=Q_{k,xxx}-2\{U,Q_{k,x}\}-\{U_x,Q_k\}+\ad_UD^{-1}_x\ad_U(Q_k),
\]
where we denote $\ad_U(V)=[U,V], \, \ad^+_U(V)=\{U,V\}=UV+VU.$
In terms of the potential $V$ defined by $U=2V_x$, this can be rewritten as
\[
 -4Q_{k+1}=R(Q_k),\quad R=D^2_x-4\ad^+_{V_x}+2D^{-1}_x\ad^+_{V_x} +4(D^{-1}_x\ad^+_{V_x})^2.
\]
 Applying this recursion operator to the identity matrix $Q_0=I$ and neglecting the integration constants, one arrives at the matrix potential-KdV hierarchy, which is obtained from the
scalar one by suitable ordering of variables in monomials:
\begin{gather}
\nonumber
 V_{t_0}=I,\qquad V_{t_1}=V_x,\qquad V_{t_2}=V_{xxx}-6V^2_x,\\
\label{vt}
 V_{t_3}=V_{xxxxx}-10\{V_x,V_{xxx}\}-10V^2_{xx}+40V^3_x,\ \dots
\end{gather}
In the scalar case any higher symmetry can be represented as a linear combination of these ones, but in the matrix case the effect of matrix integration constants leads to much larger hierarchy, which contains also nonlocal flows:
\begin{gather*}
 V_t=C,\qquad V_t=[V,C], \\
 V_t=\{V_x,C\}-[V,[V,C]]+[W,C],\quad W_x=[V,V_x],\dots
\end{gather*}
This hierarchy is not commutative, e.g. the flows $V_t=[V,C]$ and $V_\tau=[V,K]$ commute if and only if $[K,C]=0$.

The solutions of the matrix dressing chain (\ref{ALLA}) under the periodic boundary condition $U_{n+N}=CU_nC^{-1}$:
$
 A_nL_n=L_{n+1}A_n,\quad L_{n+N}=CL_nC^{-1},
 $
 satisfy the following matrix version of Novikov equations \cite{Nov}
 \begin{equation}\label{Nov}
 [A,L_n]=0,\quad A=C^{-1}A_{n+N-1}\cdots A_{n+1}A_n,
\end{equation}
which are the stationary flows of this extended KdV hierarchy.
In particular, the 1-periodic closure (\ref{fbtop}) with $\alpha=0$ can be rewritten equivalently as the Novikov equation
 \begin{equation}\label{Nov2}
 [L,A]=0,\quad A=C^{-1}(D_x+F).
\end{equation} for the corresponding Schr\"odinger operator 
$L=-D_x^2+U, \,\, U=F^2-F+B$.

However, the problem of identifying stationary solutions of KdV hierarchy and the solutions of the periodic dressing chains is not trivial even in the scalar case \cite{VS}. In the matrix situation, we have seen that the whole hierarchy is much more larger, hence the question arise, how do both languages correspond.

\section{Concluding remarks.}

Jurgen Moser was the first, who linked the integrable systems of classical mechanics (Neumann and Jacobi systems) with spectral theory and finite-gap potentials \cite{Moser, Moser2}. Our results show that classical mechanics continues to provide interesting examples for the spectral theory.

The commuting pairs of ordinary differential operators with matrix coefficients were studied by Krichever \cite{Krichever} and Grinevich \cite{Grinevich}. Dubrovin \cite{Dubrovin2} studied the algebro-geometric aspects of the matrix finite-gap case in more detail, including the related reality problems.
However, in their approach the highest coefficient is usually assumed to have different eigenvalues, so our case of Schr\"odinger operators with this coefficient being identity is very degenerate. 

A systematic study of such Schr\"odinger operators in periodic case had been started by Gesztezy and Sakhnovich in \cite{GS}, which can be served also as a review of the previous relevant results. As they have written, ``While basic aspects of inverse spectral theory for matrix-valued Schr\"odinger operators were established some time ago, finer properties such as isospectral sets (manifolds) of potentials, for instance, in the periodic or algebro-geometric finite-band cases, are still in their infancy." 

As we have seen, periodic matrix dressing chain is a source of much needed interesting explicitly solvable cases, which deserve further investigation.

\section*{Acknowledgements}
One of us (APV) had a privilege to know Valery Vasilievich Kozlov first in 1975 as an excellent tutor in Theoretical Mechanics at Moscow State University  and then as a one of the most influential experts in modern aspects of classical mechanics.
Numerous stimulating discussions with him over many years are most gratefully acknowledged.

We would also like to thank Jean-Claude Cuenin for useful discussions of the spectral aspects of the paper.


\end{document}